\documentclass[a4paper]{article}

\usepackage{arxiv}

\usepackage[utf8]{inputenc} 
\usepackage[T1]{fontenc}    
\usepackage{hyperref}       
\usepackage{url}            
\usepackage{booktabs}       
\usepackage{amsfonts}       
\usepackage{nicefrac}       
\usepackage{microtype}      
\usepackage{lipsum}		
\usepackage{graphicx}
\usepackage{natbib}
\usepackage{doi}
\usepackage{latexsym}
\usepackage{amsmath}
\usepackage{amssymb}
\usepackage[mathscr]{eucal}
\usepackage{graphicx,color}
\usepackage{mathtools}
\usepackage{stmaryrd}
\usepackage{url}
\usepackage{xypic}

\providecommand{\R}{\mathbb{R}}

\providecommand{\S}{\mathbb{S}}

\providecommand{\SO}{\mathbf{SO}}

\providecommand{\SOT}{\mathbf{SOT}}

\providecommand{\MR}{\mathbf{MR}}

\providecommand{\grpG}{\mathbf{G}}

\providecommand{\gothg}{\mathfrak{g}}

\providecommand{\so}{\mathfrak{so}}

\providecommand{\sot}{\mathfrak{sot}}

\providecommand{\calC}{\mathcal{C}}

\providecommand{\calM}{\mathcal{M}}
\providecommand{\calN}{\mathcal{N}}

\providecommand{\calP}{\mathcal{P}}

\providecommand{\calU}{\mathcal{U}}

\providecommand{\torSO}{\mathcal{SO}}

\providecommand{\torG}{\mathcal{G}}

\providecommand{\vecV}{\mathbb{V}}

\DeclareMathOperator{\diag}{diag}

\DeclareMathOperator{\Ad}{Ad}

\providecommand{\id}{\mathrm{id}} 

\providecommand{\tT}{\mathrm{T}} 
\providecommand{\td}{\mathrm{d}}
\providecommand{\tD}{\mathrm{D}}

\providecommand{\ddt}{\frac{\td}{\td t}}

\providecommand{\mr}[1]{{#1}^\circ} 

\usepackage{accents}
\makeatletter
\providecommand{\scirc}{%
    \hbox{\fontfamily{\rmdefault}\fontsize{0.4\dimexpr(\f@size pt)}{0}\selectfont{\raisebox{-0.52ex}[0ex][-0.52ex]{$\circ$}}}}

\makeatother

\mathchardef\mhyphen="2D

\usepackage{bbm}
\usepackage{tikz-cd}
\usepackage{psfrag}
\def \R {{\mathbb R}}
\def \S {{\mathbb S}}

\def \SO {{\mathbf{SO}}} 

\def \Ad {{\text{Ad}}}
\def \id {{\text{id}}}

\providecommand{\mr}{\mathring}
\providecommand{\ddt}{\frac{d}{dt}}

\newcommand{\slashedring}[1]{\overline{#1}\mathllap{\raisebox{-0.13ex}{$\mathring{\phantom{#1}}$}}}

\def \e {{\mathbf{e}}}

\usepackage{amsthm}
\newtheorem{theorem}{Theorem}
\newtheorem{proposition}{Proposition}
\newtheorem{definition}{Definition}
\newtheorem{assumption}{Assumption}
\newtheorem{lemma}{Lemma}
\newtheorem{remark}{Remark}

\title{
Exploiting polar symmetry in designing equivariant observers for vision-based motion estimation}

    \author{ \href{https://orcid.org/0000-0003-1116-7415}{\includegraphics[scale=0.06]{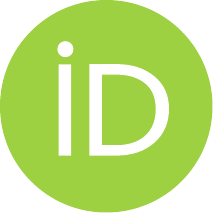}\hspace{1mm}Tarek Bouazza} \\
	I3S, CNRS, Université Côte d'Azur\\
    Sophia Antipolis, France \\
	\texttt{bouazza@i3s.unice.fr} \\
	\And
    \href{https://orcid.org/0000-0002-7803-2868}{\includegraphics[scale=0.06]{orcid.pdf}\hspace{1mm}Robert Mahony} \\
	Systems Theory and Robotics Group\\
    Australian National University\\
	ACT, 2601, Australia \\
	\texttt{Robert.Mahony@anu.edu.au} \\
    \And
	\href{https://orcid.org/0000-0002-7779-1264}{\includegraphics[scale=0.06]{orcid.pdf}\hspace{1mm}Tarek Hamel} \\
    I3S, CNRS, Université Côte d'Azur\\
    and Insitut Universitaire de France \\
    Sophia Antipolis, France \\
	\texttt{thamel@i3s.unice.fr} \\
}

\renewcommand{\headeright}{Preprint}
\renewcommand{\undertitle}{Preprint\thanks{Under review}}

\renewcommand{\shorttitle}{}

\hypersetup{
pdftitle={A template for the arxiv style},
pdfsubject={q-bio.NC, q-bio.QM},
pdfauthor={David S.~Hippocampus, Elias D.~Striatum},
pdfkeywords={First keyword, Second keyword, More},
}

\begin{document}
\maketitle

\begin{abstract}
    Accurately estimating camera motion from image sequences poses a significant challenge in computer vision and robotics.
    Many computer vision methods first compute the essential matrix associated with a motion and then extract orientation and normalized translation as inputs to pose estimation, reconstructing the scene scale (that is unobservable in the epipolar construction) from separate information.
    In this paper, we design a continuous-time filter that exploits the same perspective by using the epipolar constraint to define pseudo-measurements.
    We propose a novel polar symmetry on the pose of the camera that makes these measurements equivariant.
    This allows us to apply recent results from equivariant systems theory to estimating pose.
    We provide a novel explicit persistence of excitation condition to characterize observability of the full pose, ensuring reconstruction of the scale parameter that is not directly observable in the epipolar construction.
\end{abstract}

\section{Introduction}
	
Accurately estimating the motion of a camera from visual data is a fundamental challenge in robotics and computer vision.
The problem is key to a range of applications including visual odometry, target tracking, 3D scene reconstruction, etc.
One of the key ideas in classical computer vision used in extracting motion information from a video sequence relies on computing the so-called \emph{essential matrix} between consecutive images \cite{longuet1981computer}.
The essential matrix relates pairs of associated image points between two images taken by the same camera as it moves in the scene \cite{hartley2003multiple}.
The essential matrix can be computed from a set of at least five \cite{nister2004efficient} (or eight \cite{hartley1997defense})
matched image points, although in practice it is usually computed using nonlinear optimization over many image points  \cite{ma2001optimization,botterill2011refining}.
The essential matrix captures the essential information of the camera motion available from the visual feed.

Once the essential matrix is obtained, it can be decomposed into a rotation and a normalized translation, leading to four possible pairs of translations and rotations \cite{hartley2003multiple}.
The correct pair can be identified through a chirality check, ensuring that corresponding features are visible in front of both cameras.
The resulting pose estimation has the correct rotation and a normalized translation, that is a direction of translation.
The actual magnitude of the camera translation is linked to estimation of the scene scale and is unobservable for the image sequence alone, depending on some additional measurements and a separate motion estimation process.

Filter-based algorithms are most effective when they are posed directly on measurements, image points in this case, rather than derived information such as an essential matrix estimation.
Recent work \cite{white2020iterative} introduced an iterative algorithm that optimizes rotation and normalized translation between frames to compute relative pose based on minimizing the epipolar constraint.
In parallel, \cite{hua2020relative} proposed a deterministic Riccati observer that involves velocity measurements and uses the epipolar constraint to define a \emph{pseudo-measurement} \cite{julier2007kalman} to estimate the camera pose directly.
In \cite{hua2020relative}, and a following paper  \cite{gintrand2022uniform}, an observability analysis is provided to show that uniform observability is guaranteed provided that the translational motion of the camera is sufficiently exciting.

In this paper, we revisit the problem of designing a filter to estimate the motion of a camera by exploiting the epipolar constraint as a pseudo-measurement.
We propose a new symmetry $\SO(3) \times \SOT(3)$ for the camera pose that we term the \emph{polar symmetry} using the scaled orthogonal transformations Lie group $\SOT(3)$ \cite{van2019geometric}.
This symmetry is ideally suited to dealing with systems with unknown scale and was first proposed to model landmarks with bearing only visual measurements \cite{van2019geometric}.
For the first time (to the authors understanding) we apply the polar symmetry to the translation of the system pose.
This allows us to demonstrate equivariance of the pseudo-measurement associated with the epipolar constraint.
Moreover, we demonstrate equivariance of the system kinematics for measured velocity and from this determine a lifted system on the symmetry group.
The following observer construction is based on the equivariant filter (EqF) methodology \cite{van2020equivariant,van2020equivariant2}.
Convergence of the filter depends on excitation of the velocity signal and we provide a comprehensive theoretical analyses of observability and stability.
The resulting algorithm provides a powerful tool for tracking camera pose from visual data in the case where the velocity is measured.

The remainder of this paper is organized as follows.
The notation and preliminaries are defined in Section \ref{sec:prelims}.
The problem is formally stated in Section \ref{sec:problemform} and the proposed equivariant observer is derived in Section \ref{sec:observer}.
The observability and stability analysis is provided in Section \ref{sec:observability}.
Simulation results are presented in Section \ref{sec:simresults} followed by concluding remarks in Section \ref{sec:conclusion}.

\section{Preliminaries}
\label{sec:prelims}
\subsection{Mathematical notation}

Let $\calM$ be a smooth manifold, $\tT_{\xi}\calM$ denotes the tangent space at $\xi \in \calM$.
Given a differentiable function between smooth manifolds $h : \calM \rightarrow \calN$, its differential at $\xi_0$ is written as $$\tD_\xi \vert_{\xi_0} h(\xi): \tT_{\xi_0} \calM \rightarrow \tT_{\xi_0} \calN.$$

Let $f : \calM \rightarrow \calN$ and $g : \calN \rightarrow \calN'$ be linear maps, $f \cdot g$ denotes the composition of $f$ and $g$.

Let $\grpG$ be a matrix Lie group and $\gothg$ its Lie algebra. The group identity is denoted $\id \in \grpG$, left and right translation are written
$L_X(Y) := XY$ and $R_X(Y) := YX,$
respectively, which induce the corresponding mappings on $\gothg$, $\td L_X : \gothg \rightarrow \tT_X \grpG$ and $\td R_X : \gothg \rightarrow \tT_X \grpG$, defined by $\td L_X U := XU$ and $\td R_X U := UX$, respectively. The Adjoint map $\Ad : \grpG \times \gothg \rightarrow \gothg $ is defined by
$\Ad_X(U) := X U X^{-1} $ for any $X \in \grpG$ and $U \in \gothg$.

The Lie algebra $\gothg$ is isomorphic to a vector space $\R^{\dim\gothg}$. The \emph{wedge} $(\cdot)^\wedge_{\gothg} : \R^{\dim\gothg} \rightarrow \gothg$ and \emph{vee} $(\cdot)^\vee_{\gothg} : \gothg \rightarrow \R^{\dim\gothg} $ operators are linear isomorphisms that satisfy $(u^\vee)^\wedge = u$ for all $u \in \gothg$.
The exponential map $\exp  : \gothg \rightarrow \grpG$ defines a local diffeomorphism from a neighborhood of $0 \in \gothg$ to a neighborhood of $\id \in \grpG$, and its inverse (when defined) is the logarithmic map $\log : \grpG \rightarrow \gothg$.

The $\grpG$-torsor, denoted $\torG$, is defined as the set of elements of $\grpG$ (underlying manifold), but without the group structure.

A right group action is a smooth map $\phi : \grpG \times \calM \rightarrow \calM$ that satisfies the \emph{compatibility} and \emph{identity} properties
\begin{align*}
    \phi\left(Y, \phi(X, \xi) \right) &= \phi\left( X Y, \xi \right), &
    \phi(\id, \xi) &= \xi,
\end{align*}
for all $\xi \in \calM$ and $X\in \grpG$. For any $\xi \in \calM$ the partial maps $\phi_X : \calM \rightarrow \calM$ and $\phi_\xi : \grpG \rightarrow \calM$ are defined as $\phi_X(\xi) := \phi(X, \xi)$ and $\phi_\xi(X) := \phi(X, \xi)$, respectively.
A group action is called \emph{transitive} if for all $\xi_1, \xi_2 \in \calM$, there exists $X\in \grpG$ such that $\phi(X, \xi_1) = \xi_2$.

The special orthogonal group of 3D rotations is denoted by $\SO(3)$ with Lie algebra $\so(3)$. They are defined by
\begin{align*}
        \SO(3) &:= \{ R \in \R^{3 \times 3} \;\vert\; RR^\top = R^\top R = I_3, \det(R) = 1 \}, \\
        \mathfrak{so}(3) &:= \{ a^\times \;\vert\; a \in \R^{3} \}, \;
        a^\times := \begin{pmatrix}
            0    & -a_3 & a_2  \\
            a_3  & 0    & -a_1 \\
            -a_2 & a_1  & 0
        \end{pmatrix},
    \end{align*}
where $a^\times$ denotes the skew-symmetric matrix associated with the vector (cross) product, satisfying $a^\times b = a \times b$ for all $ b \in \R^3$. For any $a \in \R^3$, $a_{\mathfrak{so}(3)}^\wedge = a^\times$.

The scaled orthogonal transformations group denoted by $\SOT(3)$, with Lie algebra $\sot(3)$, is the direct product of $\SO(3)$ and the multiplicative positive real numbers $\MR(1) = \{r \in \R \; | \; r > 0\}$.
They are defined by
\begin{align*}
    \SOT(3) &:= \{ r Q \; | \; Q \in \SO(3), r > 0 \}, \\
    \sot(3) &:= \left\{ \left(a, b\right)^\wedge_{\sot(3)} \; | \; a \in \R^3, b \in \R  \right\}, \\
    \left(a, b\right)^\wedge_{\sot(3)} &:= a^\times + b I_3.
\end{align*}

Denote the $2$-sphere by $\S^2 := \{p \in \R^{3} \mid \vert p \vert = 1\}$, $\e_1, \e_2, \e_3$ denotes the canonical basis of $\R^3$
and $\S^+(k)$ denotes the set of positive definite $(k \times k)$ matrices.

For any $u \in \R^3 \backslash \{0\}$ the projector $\Pi_u$ that projects vectors onto the subspace of $\R^3$ orthogonal to $u$ is given by
$$ \Pi_u := I_3 - \frac{uu^\top}{|u|^2}= - \frac{u^\times u^\times}{|u|^2}. $$

\subsection{Uniform observability of linear time-varying systems}
Consider the following linear time-varying (LTV) system
\begin{equation} \label{eqn:ltv}
    \begin{cases}
        \dot{\mathbf{x}} = A(t) \mathbf{x} + B(t) \mathbf{u}, \\
        \mathbf{y} = C(t) \mathbf{x}
    \end{cases}
\end{equation}
with state $\mathbf{x} \in \R^n$, input $\mathbf{u} \in \R^l$ and output $\mathbf{y} \in \R^m$. $A, B, C$ are continuous, bounded matrix-valued functions.
\begin{definition}[\emph{Uniform observability}] \label{def:observability}
The pair $(A(t), C(t))$ of system \eqref{eqn:ltv} is called uniformly observable if there exist $\delta, \mu > 0$ such that,
for all $t\geq0$,
\begin{equation} \label{eqn:positivegramian}
        W(t, t+ \delta) := \frac{1}{\delta} \int_t^{t+\delta} \Phi^\top(s,t)C^\top(s)C(s)\Phi(s,t) ds \geq \mu I_n
    \end{equation}
where $\Phi(s, t)$ is the transition matrix associated with A(t):
$\ddt \Phi(s, t) = A(t)\Phi(s, t), \Phi(t, t) = I_n$.
\end{definition}
The matrix-valued function $W(t, t + \delta)$ is called \emph{the observability Gramian} of system \eqref{eqn:ltv}. Verifying the uniform observability directly from the Gramian is typically very challenging. The following lemma borrowed from \cite{morin2017uniform} provides a sufficient condition for uniform observability that will be instrumental in this paper.
\begin{lemma}[see \cite{morin2017uniform}] \label{morin}
If there exists a matrix-valued function
$M(\cdot)$ of dimension $(p \times n)$ $(p \geq 1)$ composed of row vectors of $N_0 = C, N_k = N_{k-1}A + \dot{N}_{k-1}, k = 1, \dots$, such that for
some positive numbers $\bar{\delta}, \bar{\mu}$ and $\forall t \geq 0$
\begin{equation}
    \frac{1}{\bar{\delta}} \int_t^{t+\bar{\delta}} |\det\left( M(s)^\top M(s)\right)| ds \geq \bar{\mu}
    \end{equation}
then $W(t, t + \delta)$ satisfies condition \eqref{eqn:positivegramian}.
\end{lemma}

\subsection{Background on epipolar geometry}
Consider a moving monocular
camera observing a 3D scene.
Let $\{ \mathring{\calC} \}$ be the initial frame of reference and $\{ \calC_t \}$ the current (camera-fixed) frame. Let $R \in \SO(3)$ represent the orientation of frame $\{ \calC_t \}$ with respect to frame $\{ \mathring{\calC} \}$ and let $x \in \R^3$ denote the translation of frame $\{ \calC_t \}$ with respect to frame $\{ \mathring{\calC} \}$ expressed in
frame $\{ \mathring{\calC} \}$.
\medbreak
Given a set of $m$ unknown landmarks that are continuously observed by the camera,  we denote $\mathring{P}_i \in \R^3$ (resp. $P_i \in \R^3$) the 3-D coordinates of the $i$-th landmark with respect to $\{ \mathring{\calC} \}$ (resp. $\{ \calC_t \}$) expressed in $\{ \mathring{\calC} \}$ (resp. $\{ \calC_t \}$). The camera measurements are modeled as the bearing measurements of the landmarks $p_i := P_i/|P_i| \in \S^2$ (resp. $\mathring{p}_i := \mathring{P}_i/|\mathring{P}_i| \in \S^2$) in the frame $\{C_t\}$ (resp. $\{\mathring{C}\}$).

Using the relations $ P_i = R^\top (\mathring{P}_i - x)$, the following epipolar constraint can be deduced
\begin{equation} \label{eq:epipolar}
    \mathring{p}_i x_d^\times R p_i = 0, \quad (i=1, \dots, m)
\end{equation}
where $x_d := x/|x|$ denotes the bearing component of the translation $x$. This constraint can be expressed in terms of the \emph{essential matrix} $E := x_d^\times R$, as
\begin{equation}
    \mathring{p}_i^\top E p_i = 0, \quad (i=1, \dots, m).
\end{equation}
Rather than estimating $E$ from the epipolar constraints and decomposing it into rotation $R$ and normalized
translation $x_d$, we propose an observer design to directly estimate the pose (i.e., $R$ and $x$), exploiting the motion of the camera.

\section{Problem formulation} \label{sec:problemform}
The pose of the camera has kinematics
\begin{align} \label{eq:kine}
    \dot{R} &= R \Omega^\times, &
    \dot{x} &= R v,
\end{align}
where $\Omega, v \in \R^3$ denote the angular and linear velocities of the camera expressed in the body-fixed frame $\{\calC_t\}$.
We assume that both velocities are measured.
\begin{assumption} \label{assumpx}
    The camera's translation with respect to the reference frame $\{\mathring{C}\}$ does not vanish; i.e. $x(t) \neq 0, \forall t \geq 0$.
\end{assumption}

Define the \emph{state space} manifold as $\calM := \torSO(3) \times \R^3 \backslash \{0\}$ with state $\xi =(R, x) \in \calM$ and the \emph{velocity space} $\mathbb{V} := \R^3 \times \R^3$ with elements $u = (\Omega, v) \in \mathbb{V}$. That is, one can write the kinematics \eqref{eq:kine} as a system
\begin{align*}
    (\dot{R}, \dot{x}) = \dot{\xi} &= f(\xi, u)=(R\Omega^\times,R v) .
\end{align*}

The corresponding bearings are regarded as parameters. We define the parameter spaces $\calP_i := \S^2 \times \S^2$, where each element is $\mathbf{p}_i = (\mathring{p}_i, p_i) \in \calP_i$, $i=1,\dots,m$. Then, the total space is $\calP := (\calP_i)^m$ with elements $\mathbf{p} = \left[\mathbf{p}_1, \dots, \mathbf{p}_m \right]^\top$.

The epipolar constraint \eqref{eq:epipolar} can be directly used to define the following measurement function $h: \calM \times \calP \rightarrow \R^m$%
\begin{align}
    h(\xi; \mathbf{p}) &:= \left( h^1(\xi; \mathbf{p}_1), \dots, h^m(\xi; \mathbf{p}_m) \right)^\top, \notag \\
    h^i(\xi; \mathbf{p}_i) &:= \mathring{p}_i^\top x_d^\times R p_i, \quad (i=1, \dots, m), \label{eq:outputs}
\end{align}
with virtual output $y := \left(y_1, \dots, y_m\right)^\top$, $y_i \equiv 0$.
This is the pseudo-measurement construction \cite{julier2007kalman} that is well known in the Kalman filtering literature.

\subsection{Symmetry and equivariance}
Define the Lie group $\grpG := \SO(3) \times \SOT(3)$ and denote an element of $\grpG$ as $X = (S,Q,r)$, $S \in \SO(3), Q \in \SO(3), r \in \MR(1)$ with group identity $\id_{\grpG} =(I_3, I_3, 1)$.

The design of an equivariant observer for the system \eqref{eq:kine} requires identifying key symmetry properties of the group $\grpG$ \cite{van2020equivariant2}. These symmetries are given below\footnote{The proofs of the following propositions are provided in the Appendix \ref{a: proofs}.}.
\begin{proposition} \label{thm:symphi}
    The mapping $\phi: \grpG \times \calM \rightarrow \calM$ defined by
    \begin{equation}
        \phi\left( (S,Q,r), (R, x) \right) := \left( Q^\top R S, r^{-1}Q^\top x \right)
    \end{equation}
    is a transitive right group action of $\grpG$ on $\calM$.
\end{proposition}
Note that while $\phi$ is transitive on the state space $\calM$, it does not exhibit transitivity on the full pose space $\mathcal{SE}(3)$ due to the disjoint nature of the orbits $\mathcal{SO}(3) \times \{0\}$ and $\calM$ under the $\phi$ action.
\begin{proposition} \label{thm:sympsi}
    The mapping $\psi: \grpG \times \vecV \rightarrow \vecV$ defined by \begin{equation}
        \psi\left( (S,Q,r), (\Omega, v) \right) := \left( S^\top \Omega, r^{-1}S^\top v \right)
    \end{equation}
    is a right group action of $\grpG$ on $\vecV$.
\end{proposition}

\begin{proposition}[\emph{Equivariance}] \label{thm:equivsys}
    The kinematics \eqref{eq:kine} are equivariant under the group actions $\phi$ and $\psi$ in the sense that
    $$ \tD \phi_X f(\xi, u) = f(\phi(X,\xi), \psi(X, u) ) $$
    for any $X \in \grpG$, $\xi \in \calM$, $u \in \vecV$.
\end{proposition}

\begin{proposition} \label{thm:symtheta}
    Define the right action $\theta: \grpG \times \calP_i \rightarrow \calP_i$ as \begin{equation}
        \theta\left( (S,Q,r), (\mathring{p}_i, p_i) \right) := \left( Q^\top \mathring{p}_i, S^\top p_i \right) .
    \end{equation}
Then, the measurement function \eqref{eq:outputs} is invariant under the actions $\phi$ and $\theta$ in the sense that
\begin{align} \label{eq:invmeasurement}
h\left(\phi(X, \xi); \theta(X, \mathbf{p})\right) = h(\xi; \mathbf{p})
\end{align}
for any $X\in \grpG$, $\xi \in \calM$, and $\mathbf{p} \in \calP$.
\end{proposition}

The inherent structural invariance \eqref{eq:invmeasurement} of the epipolar constraint \eqref{eq:epipolar} stems from the measurement equivariance under the transformation of a given pose $\xi$ by $\phi_X$ that, in turn, induces a corresponding transformation $\theta_X$ of the bearing pairs $\mathbf{p}$.

\subsection{Lift of the kinematics to the Lie algebra}
In this section, we lift the kinematics from $\calM \times \vecV$ to the Lie algebra $\gothg := \so(3) \times \sot(3)$ of $\grpG$. The existence of the lift is guaranteed by the transitive nature of $\phi$.
\begin{proposition}[\emph{Equivariant lift}] \label{thm:lift}
The smooth map $\Lambda : \calM \times \vecV \rightarrow \gothg$ defined as
\begin{align} \label{eq:lift}
\Lambda((R,x), (\Omega,v)) \notag := \left( \left(\Omega - \frac{(R^\top x)^\times v}{|x|^2} \right)^\times, - \frac{\left(x^\times Rv \right)^\times}{|x|^2}, - \frac{ x^\top Rv }{|x|^2} \right)
\end{align}
is a lift for the system \eqref{eq:kine}. That is, $\Lambda$ satisfies
\begin{align}
\tD_{X \vert \id} \phi_\xi \left[\Lambda(\xi, u) \right] = f(\xi, u).
\end{align}
In addition, the lift $\Lambda$ is equivariant, i.e.
\begin{align*}
    \mathrm{Ad}_{X^{-1}} (\Lambda(\xi, u)) = \Lambda(\phi_X (\xi),\psi_X(u)),
\end{align*}
for all $X \in \grpG$, $\xi \in \calM$ and $u \in \vecV$.
\end{proposition}

The lift \eqref{eq:lift} provides the necessary structure to construct a \emph{lifted system} on the symmetry group that projects down onto the original system dynamics via $\phi_{\mr{\xi}}: \grpG \rightarrow \calM$. Here, $\mr{\xi} = (\mr{R},\mr{x}) \in \calM$ denotes an arbitrarily fixed element of $\calM$ termed \emph{the origin}. The lifted system is given by
\begin{align} \label{eq:liftedsys}
    \dot{X} = X \Lambda\left( \phi_{\mr{\xi}}(X), u \right), \quad X(0) \in \grpG
\end{align}
written in terms of $(S,Q,r)$ as follows
\begin{align*}
\ddt (S,Q,r) = \left( S \left(\Omega - \frac{(R^\top x)^\times v}{|x|^2} \right)^\times, - Q  \frac{ \left(x^\times Rv\right)^\times}{|x|^2}, - r \frac{ x^\top Rv}{|x|^2} \right),
\end{align*}
with $(R, x) = \phi_{(\mr{R},\mr{x})}\left((S, Q, r)\right)$.
\section{Equivariant observer design} \label{sec:observer}
This section presents an \emph{equivariant observer} designed on the symmetry group $\grpG$ and uses the lifted system \eqref{eq:liftedsys} as its internal model. It follows the equivariant filter (EqF) design approach as outlined in \cite{van2020equivariant2} and \cite{mahony2022observer}.
\medbreak
Let $\hat{X} = (\hat{S}, \hat{Q}, \hat{r}) \in \grpG$ be the observer state with kinematics
\begin{equation} \label{eq:equivobserver}
   \dot{\hat{X}} := \hat{X}\Lambda(\phi_{\mr{\xi}}(\hat{X}), u) + \Delta_t \hat{X}, \quad \hat{X}(0) \in \grpG
\end{equation}
where $\Delta_t \in \gothg$ is the correction term to be determined. 
\subsection{Origin choice and local coordinates}
The choice of the origin of the translational component $\mr{x} \in \R^3 \backslash \{0\}$ can be arbitrary. Then, without loss of generality, let $(\mr{R}, \mr{x}) = (I_3, \e_3) \in \calM$ be the chosen state origin. At any time $t$,
the EqF state estimate is given by
\begin{align}
\hat{\xi}(t) = \phi_{\mr{\xi}}(\hat{X}(t)) = \left( \hat{Q}^\top \hat{S}, \hat{r}^{-1}\hat{Q}^\top \e_3 \right).
\end{align}

Let $e = (e_R, e_t) := \phi_{\hat{X}^{-1}}(\xi) \in \calM$ denote the global \emph{equivariant error}. The EqF is derived by linearizing the dynamics of $e$ about $\mr{\xi}$ which requires a chart of local coordinates for the state.

The following  \emph{polar} parameterization was previously  introduced in \cite{van2023eqvio}:
\begin{align*}
\zeta_{\SOT(3)}(q) &:= \begin{pmatrix}
   \arccos\left(\frac{q_3}{|q|}\right) \frac{q_2}{|\e_3 \times q|} \\
   \arccos\left(\frac{q_3}{|q|}\right) \frac{-q_1}{|\e_3 \times q|}
 \\
 -\log(|q|)
\end{pmatrix}, \notag \\
 \zeta^{-1}_{\SOT(3)}(z) &:= \exp_{\SOT(3)}\left( \left(\begin{pmatrix}
     z_1 \;\; z_2  \;\; 0 \;\; z_3
 \end{pmatrix}^\top \right)^\wedge_{\sot(3)} \right)^{-1} \e_3.
\end{align*}
It provides normal coordinates for $\R^3$ about $\e_3$ with respect to the right action of $\SOT(3)$ defined in Lemma \ref{thm:symphi}.
Define the map $\vartheta: \calU_{\mr{\xi}} \subset \calM \rightarrow \R^6$
\begin{align}
\vartheta(e) &:= \left( \left(\log_{\SO(3)} (e_R) \right)^\vee_{\so(3)}, \zeta_{\SOT(3)}(e_t) \right)
\end{align}
to be the coordinate chart for $\calM$ about $\mr{\xi}$, where $\calU_{\mr{\xi}}$ is a large neighborhood of $\mr{\xi}$. The map $\vartheta$ provides normal coordinates for $\calM$ about $\mr{\xi}$ with respect to the action $\phi$,
its inverse is given by
\begin{equation}
\vartheta^{-1}(\varepsilon) := \left( \exp_{\SO(3)} \left( \varepsilon_R^\times \right), \zeta^{-1}_{\SOT(3)}(\varepsilon_t)\right).
\end{equation}

Then, linearizing the error dynamics about $\mr{\xi}$ yields \cite{van2020equivariant2}
\begin{align*}
    \dot{\varepsilon} &= A_t \varepsilon + O(|\varepsilon|^2), \\
    A_t &= \tD_{e \vert \mr{\xi}} \vartheta(e) \cdot \tD_{E \vert \id} \phi_{\mr{\xi}}(E) \cdot \tD_{e \vert \mr{\xi}} \Lambda (e,\psi_{\hat{X}^{-1}}(u)) \cdot \tD_{\varepsilon \vert 0} \vartheta^{-1}(\varepsilon).
\end{align*}
\medbreak
Define the output residual $\tilde{y} :=  h(\xi, \mathbf{p}) - h(\hat{\xi}, \mathbf{p}) \in \R^m$.
Then using \eqref{eq:invmeasurement}, one has
\begin{align*}
        \tilde{y} &= h(\xi, \mathbf{p}) - h(\hat{\xi}, \mathbf{p}), \\
        &= h\left(\vartheta^{-1}(\varepsilon), \theta_{\hat{X}^{-1}}(\mathbf{p})\right)  - h(\mr{\xi}, \theta_{\hat{X}^{-1}}(\mathbf{p})).
\end{align*}
with $h^i(\mr{\xi}, \theta_{\hat{X}^{-1}}(\mathbf{p}_i)) = \hat{\slashedring{p}}_i^\top \e_3^\times \hat{\bar{p}}_i$, $(\hat{\slashedring{p}}_i, \hat{\bar{p}}_i) = \theta_{\hat{X}^{-1}}(\mathring{p}_i, p_i) = (\hat{Q} \mathring{p}_i, \hat{S} p_i)$, $i=1, \dots, m$.

The linearized output about $\varepsilon = 0$ is
\begin{align*}
    \tilde{y} &= C_t \varepsilon + O(|\varepsilon|^2), \\
    C_t &= \tD_{e \vert \mr{\xi}} h(e; \theta_{\hat{X}^{-1}}(\mathbf{p}_i))  \cdot \tD_{\varepsilon \vert 0} \vartheta^{-1} (\varepsilon).
\end{align*}
The resulting state matrix $A_t$ and output matrix $C_t$ are
\begin{equation} \label{eqn:linearizedmat}
\begin{cases}
    A_t = \begin{pmatrix}
    (\e_3^\times \hat{\mathring{v}})^\times & 0_{3,2} & 0_{3,1} \\
    0_{2,3} &  \hat{\mathring{v}}_3 I_2 & \begin{bmatrix}
        \hat{\mathring{v}}_2 \\ -\hat{\mathring{v}}_1
    \end{bmatrix} \\
    0_{1,3} & \begin{bmatrix}
        -\hat{\mathring{v}}_2 & \hat{\mathring{v}}_1
    \end{bmatrix} & \hat{\mathring{v}}_3
\end{pmatrix} \in \R^{6\times 6}, \\
C_t = \begin{pmatrix}
        \e_3^\top \hat{\slashedring{p}}_1^\times\hat{\bar{p}}_1^\times  & -\e_2^\top \hat{\slashedring{p}}_1^\times \hat{\bar{p}}_1  & \e_1^\top \hat{\slashedring{p}}_1^\times \hat{\bar{p}}_1 & 0 \\
        \vdots & \vdots & \vdots & \vdots \\
        \e_3^\top \hat{\slashedring{p}}_m^\times\hat{\bar{p}}_m^\times & -\e_2^\top \hat{\slashedring{p}}_m^\times \hat{\bar{p}}_m  & \e_1^\top \hat{\slashedring{p}}_m^\times \hat{\bar{p}}_m & 0  \\
\end{pmatrix} \in \R^{m\times 6}.
\end{cases}
\end{equation}
with $\hat{\mathring{v}} = \hat{r} \hat{S} v$.
Then, the correction term $\Delta_t$ is given by
\begin{align} 
\Delta_t &:= \left(\Sigma C_t^\top N_t^{-1}\tilde{y}\right)^\wedge_\gothg,
\\ \label{eq: riccati}
\dot{\Sigma} &:= A_t \Sigma + \Sigma A_t^\top + M_t - \Sigma C_t^\top N_t^{-1} C_t\Sigma, 
\end{align}
where $\Sigma \in \S^+(6)$ is the Riccati gain with initial
value $\Sigma(0) = \Sigma_0$, $M_t \in \S^+(6)$ and $N_t \in \S^+(m)$ are
continuous matrix-valued functions. In a stochastic setting (Kalman filtering), $M_t$ and $N_t$ are interpreted as covariance matrices of additive noise on the state and output, respectively \cite{van2020equivariant2}.

\section{Observability and stability analysis} \label{sec:observability}
In this section, we outline the necessary conditions to ensure the local exponential stability of the linearized origin error of the observer \eqref{eq:equivobserver}. According to \cite[Corollary 3.2]{hamel2017riccati}, the exponential stability relies on the uniform observability in the sense of Definition \ref{def:observability} of the pair $(\mathring{A}_t, \mathring{C}_t)$ obtained by setting $\hat{X}(t) = X(t)$ in the expressions of $(A_t,C_t)$ in \eqref{eqn:linearizedmat}.

In view of \eqref{eqn:linearizedmat}, the expression of $\mathring{A}_t$ corresponds to substituting $\hat{\mathring{v}}$ by $\mathring{v} = r S v$ in $A_t$. Let $(\slashedring{p}_i, \bar{p}_i) = (Q \mathring{p}_i, S p_i)$, $\mathring{C}_t$ can be expressed as $\mathring{C}_t = \bar{C}_t L$,
    with $L = \begin{pmatrix}
           I_5 & 0_{5,1}
        \end{pmatrix}$ and
        \begin{equation} \label{eqn:cbar}
        \bar{C}_t = \begin{pmatrix}
            \e_3^\top \slashedring{p}_1^\times\bar{p}_1^\times  & -\e_2^\top \slashedring{p}_1^\times \bar{p}_1  & \e_1^\top \slashedring{p}_1^\times \bar{p}_1   \\
            \vdots & \vdots & \vdots \\
            \e_3^\top \slashedring{p}_m^\times\bar{p}_m^\times & -\e_2^\top \slashedring{p}_m^\times \bar{p}_m  & \e_1^\top \slashedring{p}_m^\times \bar{p}_m  \\
        \end{pmatrix}
        \end{equation}
        
        Let $\mathring{\Phi}$ denote the state transition matrix associated with $\mathring{A}_t$. The observability Gramian associated with $(\mathring{A}_t,\mathring{C}_t)$ is
\begin{equation} \label{eqn:gramian0}
    W(t, t+ \delta) = \frac{1}{\delta} \int_t^{t+\delta} \mathring{\Phi}^\top(s,t)\mathring{C}(s)^\top \mathring{C}(s)\mathring{\Phi}(s,t) ds 
\end{equation}
\begin{assumption} \label{assump1}
    There are at least five landmarks $(m \geq 5)$ that are uniformly non-collinear, such that there exists at least a triplet $\mathring{p}_1, \mathring{p}_2, \mathring{p}_3$ and $c > 0$ that satisfy $(\mathring{p}_1 \times \mathring{p}_2)^\top \mathring{p}_3 \geq c$. Additionally, if these landmarks are  positioned across a \textbf{horopter curve} (the intersection of a circular cylinder and an elliptic cone \cite{hamel2017riccati}), the origin of the reference frame $\{\mathring{\calC}\}$ is uniformly distant from the horopter origin.
    \end{assumption}
    
    \begin{lemma} \label{Cinvertible}
        If Assumption \ref{assump1} holds and there exists $\epsilon > 0$ such that for all $t\geq 0$, $|x| \geq \epsilon$,
        then $\bar{C}_t$ given in \eqref{eqn:cbar} is full rank and $\bar{H}_t := \bar{C}_t^\top\bar{C}_t$ is invertible and well-conditioned.
    \end{lemma}
    
    \begin{proof}
    To prove that $\bar{C}_t$ is full rank under Assumption \ref{assump1}, it is sufficient to show that the equation $\bar{C}_t w = 0$, with $w = \left(w_1^\top, w_2^\top\right)^\top, w_1 \in \R^3, w_2 \in \R^2$, implies that the unique solution is $w = 0$. Then,
    \begin{align*}
        & \quad \bar{C}_t w = 0 \\
        \Rightarrow & \quad \e_3^\top \slashedring{p}_i^\times\bar{p}_i^\times w_1 -\e_2^\top \slashedring{p}_i^\times \bar{p}_i w_{2,1}  + \e_1^\top \slashedring{p}_i^\times \bar{p}_i w_{2,2} = 0,
    \end{align*}
    The relations $\mathring{P}_i = R P_i + x$  yield $\bar{p}_i = \frac{1}{|P_i|}(|\mathring{P}_i| \slashedring{p}_i - |x| e_3)$. Let $\underline{w}_2 = \begin{pmatrix} w_2 \\ 0 \end{pmatrix} \in \R^3$, it follows that for all $i=1, \dots, m$
    \begin{align}
    & \Rightarrow \; \e_3^\top \slashedring{p}_i^\times \left(|\mathring{P}_i| \slashedring{p}_i^\times w_1 + |x| \e_3^\times \left( \underline{w}_2 - w_1 \right)\right) = 0, \notag \\
    \begin{split}
    & \Rightarrow \; \exists \alpha_i \in \R \; \text{such that} \; \\
    & \quad\quad\quad\quad
    \slashedring{p}_i^\times \left(|\mathring{P}_i| \slashedring{p}_i^\times w_1 + |x| \e_3^\times \left( \underline{w}_2 - w_1 \right) \right)
    = \alpha_i \slashedring{p}_i^\times \e_3,
    \end{split}
    \notag \\
    & \Rightarrow \; \slashedring{p}_i^\times \left(|\mathring{P}_i| \slashedring{p}_i^\times w_1 + \bar{w} \right) = 0, \label{eqn:sysequations}
    \end{align}
    with $\bar{w} = |x|\e_3^\times \left(  \underline{w}_2 - w_1 \right) - \alpha_i \e_3$. It is clear that $w_2$ cannot be arbitrary since $|x|\geq \epsilon$ for all $t\geq 0$.
    
    The remaining proof, ensuring that $w = 0$ is the unique solution and hence $\bar{C}_t$ is full rank, is outlined in \cite[Section IV-B]{hamel2017riccati}.
    
    Using the fact that $\bar{C}_t$ is composed of bounded elements ($\e_j, \slashedring{p}_i, \bar{p}_i$ are elements of $\S^2$), one ensures  that $\bar{C}_t$ is bounded and hence the maximal eigen value $\lambda_{\max}(\bar{H}_t)$ is also bounded.
    Now, the uniformity stated in Assumption \ref{assump1} implies the existence of $\bar{\epsilon} > 0$ such that for all $t\geq 0$, $\lambda_{\min}(\bar{H}_t) \geq \bar{\epsilon}$. From there, one ensures that the condition number of $\bar{H}_t=\lambda_{\max}/\lambda_{\min}$ is bounded and hence $\bar{H}_t$ is well-conditioned.
    \end{proof}
\begin{theorem} \label{thm:observability} 
If Assumption \ref{assump1} holds and the linear velocity $v(t)$ is ``persistently exciting'' in the sense that there exists $\delta, \mu > 0$ such that $\forall t \geq 0$
\begin{equation} \label{eqn:pecond}
    \frac{1}{\delta} \int_t^{t+\delta} \frac{(Rv)^\top \Pi_{x} Rv}{|x|^2} ds \geq \mu.
\end{equation}
Then, the matrix pair $(\mathring{A}_t, \mathring{C}_t)$ is uniformly observable. Consequently, $\Sigma(t)$ and $\Sigma^{-1}(t)$ are uniformly bounded and the origin $\varepsilon(t) = 0$ is locally exponentially stable.
\end{theorem}

\begin{proof}
    Taking into account the Gramian \eqref{eqn:gramian0} and given that $\bar{H}_t$ is well-conditioned under Assumption \ref{assump1}, one deduces that $W(t, t + \delta) \geq \lambda_{\mathrm{min}}(\bar{H}_t)\bar{W}(t, t + \delta)$, where $\bar{W}(t, t+ \delta)$ is the observability Gramian of the pair $(\mathring{A}_t, L)$ given by
    \begin{equation} \label{eqn:partialgramian}
         \bar{W}(t, t+ \delta) = \frac{1}{\delta} \int_t^{t+\delta} \mathring{\Phi}^\top(s,t)L^\top L\mathring{\Phi}(s,t) ds 
     \end{equation}
    Therefore, ensuring the uniform observability of the pair $(\mathring{A}_t, \mathring{C}_t)$ amounts to that of $(\mathring{A}_t, L)$.
    We show thereafter that condition \eqref{eqn:pecond} is sufficient for \eqref{eqn:partialgramian} to satisfy \eqref{eqn:positivegramian}. Applying Lemma \ref{morin}, we define the matrix-valued function $\bar{M}(t) := \left(N_0^\top \; N_1^\top \right)^\top$, with $N_0 = L$, $N_1(t) = L \mathring{A}_t$.
    One verifies using the expressions of $\mathring{A}_t$ and $L$ that
    \begin{equation} \label{eqn:matrixM}
        \bar{M}(t) = \begin{pmatrix}
           I_3 & 0_{3,2} & 0_{3,1} \\
            0_{2,3} & I_2 & 0_{2,1} \\
            (\e_3^\times \mathring{v})^\times & 0_{3,2} & 0_{3,1} \\
            0_{2,3} & \mathring{v}_3 I_2 & \begin{bmatrix}
            \mathring{v}_2 \\ -\mathring{v}_1
        \end{bmatrix}
        \end{pmatrix}
    \end{equation}
    Thus, one deduces
    \begin{align*}
        &    \det\left( \bar{M}(t)^\top \bar{M}(t) \right) \\ &= \det\begin{pmatrix}
                I_3 - (\e_3^\times \mathring{v})^\times (\e_3^\times \mathring{v})^\times & 0_{3,2} & 0_{3,1} \\
                0_{2,3} & (1 + \mathring{v}_3^2) I_2 & \mathring{v}_3 \begin{bmatrix}
                    \mathring{v}_2  \\ -\mathring{v}_1 
                \end{bmatrix} \\
                0_{1,3} & \mathring{v}_3 \begin{bmatrix}
                \mathring{v}_2 & -\mathring{v}_1
            \end{bmatrix} &  |\e_3^\times \mathring{v}|^2
            \end{pmatrix} \\ 
            &= \det\left(I_3  + |\e_3^\times \mathring{v}|^2 \Pi_{\e_3^\times \mathring{v}}  \right) (1+\mathring{v}_3^2)|\e_3^\times \mathring{v}|^2 \\
            &= (1 + |\e_3^\times \mathring{v}|^2)^2 (1+\mathring{v}_3^2) |\e_3^\times \mathring{v}|^2
            \geq |\e_3^\times \mathring{v}|^2. 
        \end{align*}
        Since $\e_3 = r Q x$, $\mathring{v} = r S v$, $R= Q^\top S$ and $|x| = r^{-1}$, this can be equivalently rewritten as
    $$ \det\left( \bar{M}(t)^\top \bar{M}(t) \right) \geq \frac{1}{|x|^2}(Rv)^\top \Pi_{x} Rv. $$ 
    Then, using \eqref{eqn:pecond}, one concludes that $(\mathring{A}_t, L)$, and hence $(\mathring{A}_t, \mathring{C}_t)$, is uniformly observable. The remainder of the proof follows from Theorem 3.1 and Corollary 3.2 in \cite{hamel2017riccati}.
\end{proof}

\begin{remark}
The persistence of excitation condition \eqref{eqn:pecond} is compromised when the linear velocity of the camera motion in the reference frame aligns with the position vector (linear motion along the straight line passing through the centers of the camera frames).
\end{remark}

\section{Simulation results}\label{sec:simresults}

To evaluate the performance of the proposed observer, we present simulation results based on a discretized version.
The simulated scenario involves five landmarks satisfying Assumption \ref{assump1} and is divided into three stages of motion to illustrate the observability conditions outlined in section \ref{sec:observability}:

(a) During $t \in [0,1]$, the camera maintains a static position with the centers of the camera frames aligned on the $z$-axis, such that $x(t) = (0, 0, 1)^\top$.

(b) During $t \in [1,4]$, the camera oscillates periodically along the $z$-axis with velocity $\left(0, 0, \frac{1}{2}\sin(\pi t) \right)^\top$ in the frame $\{\mathring{\calC}\}$, violating the persistence of excitation condition \eqref{eqn:pecond}.

(c) During $t \in [4,8]$, it traces a circular path in the $xy$-plane with velocity $\left(\sin(\pi t), -\cos(\pi t), 0\right)^\top$.
 This motion satisfies the \emph{p.e.} condition \eqref{eqn:pecond} and leads to full observability.

In phases (b) and (c), the angular velocity is set as $\Omega(t) = \frac{\pi}{20}\left(\cos(t), 2\cos(2t), 5\cos(2t) \right)^\top$.
The initial estimates are set as follows: $\hat{S}(0)$ corresponds to errors in roll, pitch, and yaw of $45(\mathrm{deg})$, $\hat{Q}(0)$ corresponds to errors in roll and pitch of $30(\mathrm{deg})$ and $\hat{r}(0) = 0.5(\mathrm{m})$.
The chosen observer parameters are: $\Sigma_0 = \diag(I_5, 5)$, $N_t = 0.01 I_5$ and $M_t = \diag(0.01I_5, 0.01\alpha)$, with $\alpha = (\hat{R}v)^\top \Pi_{\hat{x}} \hat{R}v$. This choice of $M_t$ ensures that $\Sigma(t)$ remains well-conditioned even when the velocity is not persistently exciting.

Fig. \ref{fig:1} shows the convergence of the estimation errors of the position direction $x_d = x/|x|$ and orientation $R$. Fig. \ref{fig:2} illustrates the estimation error of the range $|x|$ and the Lyapunov function $\mathcal{L}(t) := \varepsilon^\top \Sigma^{-1} \varepsilon$ of the linearized state error across the three stages of motion. 

The estimation errors for the position direction and the orientation rapidly converge to zero in phase (a), without requiring motion as sufficient landmarks are satisfying Assumption \ref{assump1}. Whereas, the range component does not converge in phases (a) and (b) because the velocity signal is not persistently exciting, resulting in a loss of uniform observability. 
When the persistence of excitation is fulfilled in phase (c), both the estimation error for the range and the Lyapunov function value rapidly converge to zero. These results confirm the exponential stability of the full state error. 

\begin{figure}[h]
        \centering
        \includegraphics[scale=.45]{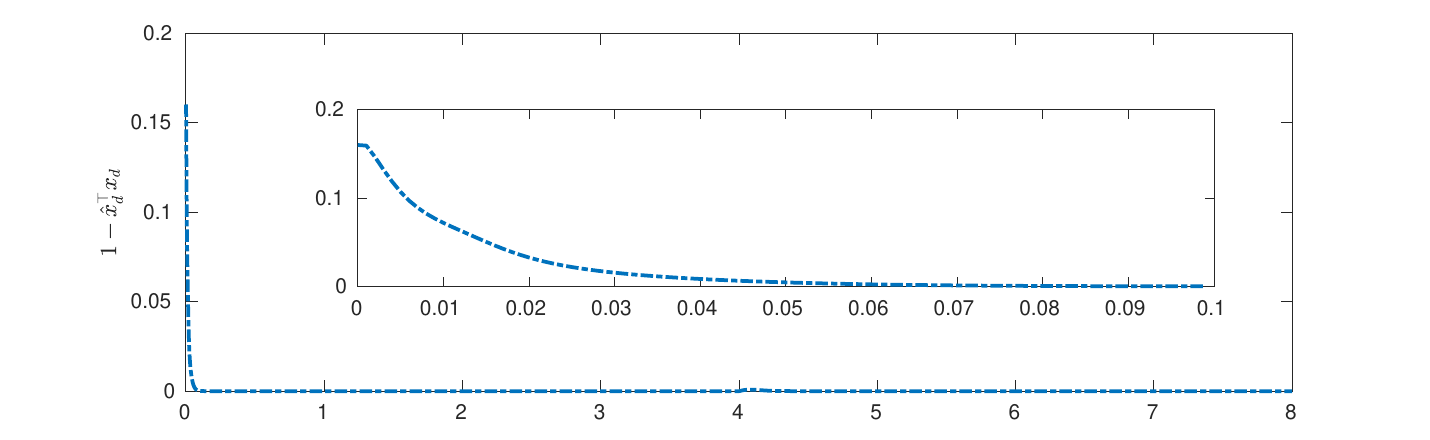}
        \includegraphics[scale=.45]{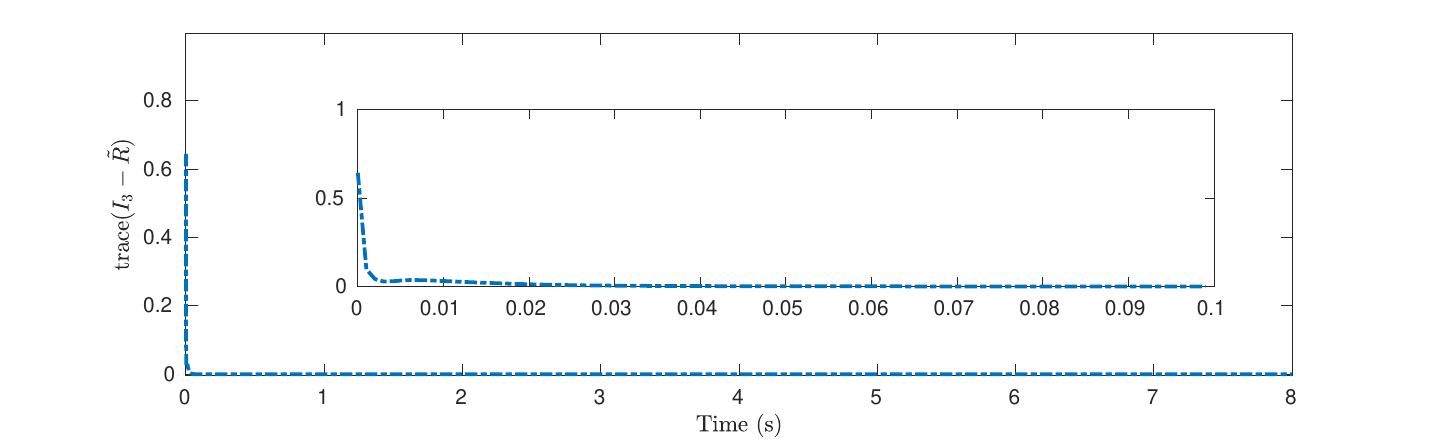}
        \caption{Estimation errors on the position direction and orientation.}
        \label{fig:1}
    \end{figure}
    \begin{figure}[h]
        \centering
        \includegraphics[scale=.45]{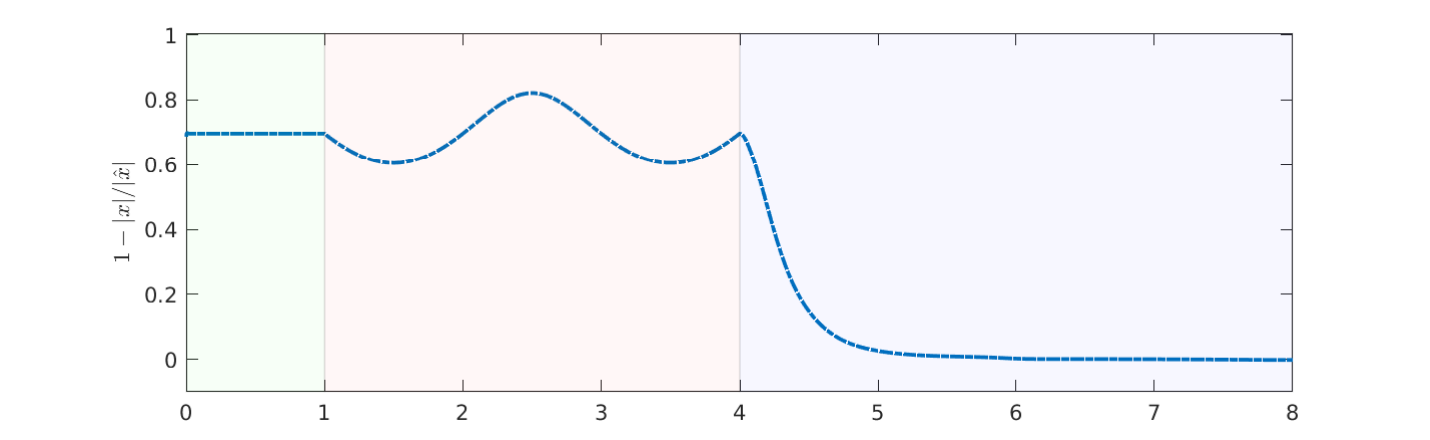}
        \includegraphics[scale=.45]{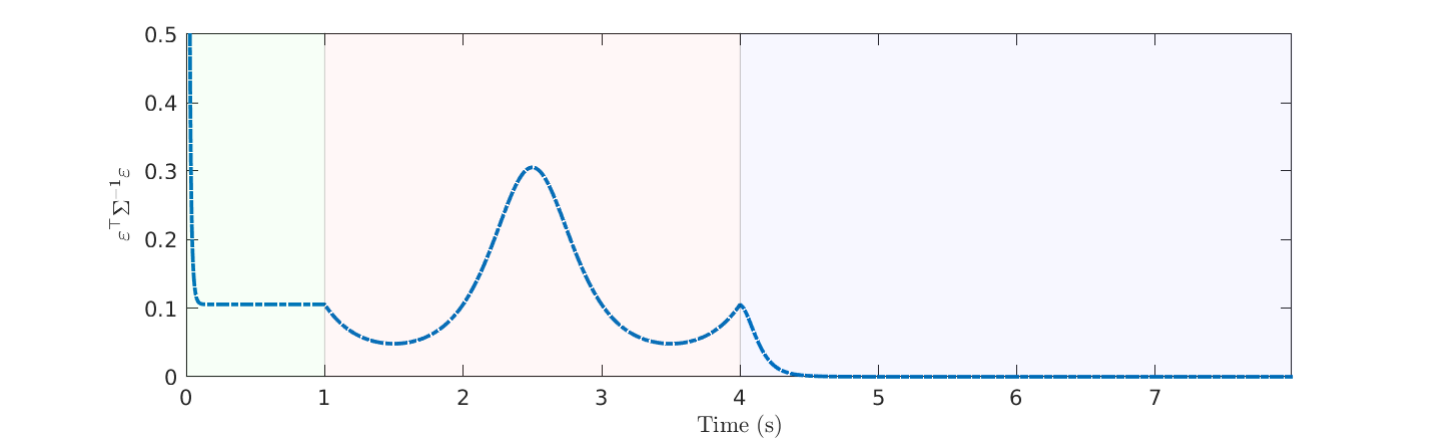}
        \caption{Estimation error on the position range and the Lyapunov function value during motion phases: (a) in green, (b) in red and (c) in blue.}
        \label{fig:2}
    \end{figure}

    \section{Conclusions} \label{sec:conclusion}

    We presented an equivariant observer design to estimate relative pose using epipolar geometry and velocity measurements. The approach is based on a novel \emph{polar symmetry} employed to parametrize 3D pose, efficiently decoupling the bearing and range components and naturally aligning with the scale-invariance of the epipolar constraint.
    Comprehensive observability and stability analyses were carried out in support of the proposed observer establishing an explicit persistence of excitation condition to ensure the uniform observability of the range component.
    The provided simulations validate the theoretical results and illustrate the performance of the proposed approach.


    \section*{Acknowledgment}
    This work has been supported by the French government, through the EUR DS4H Investments in the Future project managed by the National French Agency (ANR) with the reference number ANR-17-EURE-0004, the ANR-ASTRID Project ASCAR, the Franco-Australian International Research Project “Advancing Autonomy for Unmanned Robotic Systems” (IRP ARS) and the Australian Research Council through Discovery Grant DP210102607 ``Exploiting the Symmetry of Spatial Awareness for 21st Century Automation''.


\bibliographystyle{unsrtnat}
\bibliography{ref}

\appendix


\section{Proofs}
\label{a: proofs}
\begin{proof}[Proof of Proposition \ref{thm:symphi}] 
The identity property $\phi\left( (I_3, I_3, 1), (R, x) \right) = (R, x) $ is straightfoward to verify for any $(R,x) \in \calM$. Let $(R,x) \in \calM$ and $(S_1, Q_1, r_1),(S_2, Q_2, r_2) \in \grpG$ be arbitrary. Then
    \begin{align*}
        \phi\left( (S_2, Q_2, r_2), \phi\left( (S_1, Q_1, r_1), (R, x) \right) \right) 
        &= \phi\left((S_2, Q_2, r_2), \left(Q_1^\top R S_1, r_1^{-1} Q_1^\top x\right) \right), \\
        &= \left( Q_2^\top Q_1^\top R S_1 S_2, r_2^{-1}Q_2^\top\left(r_1^{-1} Q_1^\top x\right) \right), \\
        &= \left( (Q_1 Q_2)^\top R S_1 S_2, (r_1 r_2)^{-1}(Q_1 Q_2)^\top x \right), \\
        &= \phi\left( (S_1 S_2, Q_1 Q_2, r_1 r_2), (R,x) \right), \\
        &= \phi\left( (S_1, Q_1, r_1)(S_2, Q_2, r_2), (R,x) \right),
    \end{align*}
so it satisfies compatibility and it follows that $\phi$ is a right action. 
The transitivity follows from the property that $\phi((S,Q,r), (I_3, x_0)) = (S^\top Q, r^{-1} Q^\top x_0)$ and hence it is straightforward to see any point in $\calM$ can be reached from $(I_3, x_0)$ by suitable construction of an element of $\grpG$. This completes the proof.
\end{proof}

\begin{proof}[Proof of Proposition \ref{thm:sympsi}] 
The identity property $\psi\left( (I_3, I_3, 1), (\Omega, v) \right) = (\Omega, v) $ is straightfoward to verify for any $(\Omega,v) \in \vecV$. Let $(\Omega,v) \in \vecV$ and $(S_1, Q_1, r_1),(S_2, Q_2, r_2) \in \grpG$ be arbitrary. Then
    \begin{align*}
        \psi\left( (S_2, Q_2, r_2), \psi\left( (S_1, Q_1, r_1), (\Omega, v) \right) \right) 
        &= \psi\left((S_2, Q_2, r_2), \left( S_1^\top \Omega, r_1^{-1}S_1^\top v \right) \right), \\
        &= \left( S_2^\top S_1^\top \Omega, r_2^{-1}S_2^\top\left(r_1^{-1}S_1^\top v\right) \right), \\
        &= \left( (S_1 S_2)^\top \Omega, (r_1r_2)^{-1}(S_1 S_2)^\top v \right), \\
        &= \psi\left( (S_1 S_2, Q_1 Q_2, r_1 r_2), (\Omega,v) \right), \\
        &= \psi\left( (S_1, Q_1, r_1)(S_2, Q_2, r_2), (\Omega,v) \right),
    \end{align*}
so it satisfies compatibility. This demonstrates that $\psi$ is a right action as required.
\end{proof}

\begin{proof}[Proof of Proposition \ref{thm:equivsys}]
Let $(S, Q, r) \in \grpG$, $(R, x) \in \calM$ and
$(\Omega, v) \in \vecV$ be arbitrary. Note that $\phi$ is linear in $R$ and $x$, so $\tD \phi_X$ acts on $f(\xi, u)$ the same way that $\phi_X$ acts on $\xi$. Therefore, one has
\begin{align*}
    \tD \phi_{(S,Q,r)} f\left((R,x), (\Omega,v)\right) 
    &= \tD \phi_X \left( R \Omega^\times, Rv  \right), \\
    &= \left( Q^\top R \Omega^\times S, r^{-1}Q^\top R v \right), \\
    &= \left( Q^\top R S S^\top \Omega^\times S, r^{-1} Q^\top R S S^\top v \right), \\
    &= \left( (Q^\top R S) (S^\top \Omega)^\times, (Q^\top R S)\left(r^{-1} S^\top v\right) \right), \\
    &= f \left( \phi((S,Q,r), (R,x)), \psi((S,Q,r), (\Omega, v)) \right).
\end{align*}
This proves that the kinematics \eqref{eq:kine} are equivariant under the symmetries $\phi$ and $\psi$.
\end{proof}

\begin{proof}[Proof of Proposition \ref{thm:symtheta}] 
    It is straightforward to verify that $\theta$ is a right action of $\grpG$ on $\calP$.
    To show the invariance of $h$, it is sufficient to show that the the component functions $h^i$ are invariant under the actions $\phi$ and $\theta$. Let $(S, Q, r) \in \grpG$, $(R, t) \in \calM$ and $(\mathring{p}_i, p_i) \in \calP_i$ be arbitrary. Then,
    \begin{align*}
        h^i\left(\phi\left((S,Q,r), (R, x)\right); \theta\left((S,Q,r),(\mathring{p}_i, p_i)\right) \right) 
        &= h^i\left(\left( Q^\top R S, r^{-1}Q^\top x\right);  \left(Q^\top \mathring{p}_i, S^\top p_i \right) \right), \\
        &= \left(Q^\top \mathring{p}_i\right)^\top \frac{\left(r^{-1}Q^\top x\right)^\times}{|r^{-1}Q^\top x|} (Q^\top R S)(S^\top p_i), \\
        &= \mathring{p}_i^\top Q Q^\top x_d^\times Q Q^\top R S S ^\top p_i, \\
        &=  \mathring{p}_i^\top x_d^\times R p_i = h^i((R,x); (\mathring{p}_i, p_i)).
    \end{align*}
    It follows that
    \begin{align*}
        &h\left(\phi\left((S,Q,r), (R, x)\right); \theta\left((S,Q,r),\mathbf{p}\right) \right) \\ 
        &= \left(h^1\left(\phi\left((S,Q,r), (R, x)\right); \theta\left((S,Q,r),\mathbf{p}_1\right) \right), \dots, h^m\left(\phi\left((S,Q,r), (R, x)\right); \theta\left((S,Q,r),\mathbf{p}_m\right) \right) \right)  \\
        &= \left( h^1\left((R, x); \mathbf{p}_1 \right), \dots, h^m\left((R, x); \mathbf{p}_m  \right) \right) \\
        &= h\left((R, x); \mathbf{p}  \right).
    \end{align*}
    This shows that, indeed, $h$ is invariant under the actions $\phi$ and $\theta$.
\end{proof}

\begin{proof}[Proof of Proposition \ref{thm:lift}]
Recall that $\phi((S,Q,r),(R, x) = \left(Q^\top RS,r^{-1}Q^\top x\right)$.
To find $\tD_{(S,Q,r) \vert (I_3,I_3,1)} \phi_\xi$ first choose $s \in \mathfrak{so}(3)$, $q \in \so(3)$ and $b \in \mathfrak{mr}(1)$, and then evaluate $\phi$ applied to $S=e^{ts}$, $Q=e^{tq}$ and $r = e^{tb}$:
	\begin{align*}
		\tD_{(S,Q,r) \vert (I_3,I_3,1)} \phi_\xi (s,q,b) 
  &= \ddt \phi((e^{ta},e^{tq},e^{tb}), (R,x)) \mid_{t=0}, \\
  &= \ddt  (e^{-tq}Re^{ts}, e^{-tb}e^{-tq}x) \mid_{t=0}, \\
  &= (-q R + Rs, -(bI_3 + q)x), 
	\end{align*}
and so
	\begin{align*}  
        \tD_{X\vert \id} \phi_\xi [\Lambda(\xi, u)] &= \left( \left( \frac{x^\times Rv}{|x|^2} \right)^\times R + R\left(\Omega - \frac{(R^\top x)^\times v}{|x|^2} \right)^\times , \frac{1}{|x|^2}\left( (x^\top Rv) I_3 + (x^\times Rv)^\times\right)x \right),
        \\
        &= \left( R R^\top \left( \frac{x^\times Rv}{|x|^2} \right)^\times R + R\left(\Omega - \frac{(R^\top x)^\times v}{|x|^2} \right)^\times , \frac{1}{|x|^2}\left( \langle x, Rv\rangle I_3 + x^\times (Rv)^\times - (Rv)^\times x^\times \right)x \right),\\
        &= \left(R \left( \frac{ R^\top x^\times R v }{|x|^2} + \Omega - \frac{(R^\top x)^\times v}{|x|^2} \right)^\times , \frac{1}{|x|^2}\left(  xx^\top Rv - x^\times x^\times Rv \right) \right), \\
        &= \left( R \left( \frac{ (R^\top x)^\times v }{|x|^2} + \Omega - \frac{(R^\top x)^\times v}{|x|^2} \right)^\times , \left( \frac{xx^\top}{|x|^2} - \frac{x^\times x^\times}{|x|^2} \right) Rv  \right),\\
        &= \left(R \Omega^\times , \left( \frac{xx^\top}{|x|^2} + \Pi_x \right) Rv \right),\\
        &= \left( R \Omega^\times , Rv \right) = f\left( \xi, u \right).
	\end{align*}
As required. 
Now to show that the lift $\Lambda$ is equivariant
	\begin{align*}
		&\mathrm{Ad}_{X^{-1}}(\Lambda(\xi,u)) \\
        &= \mathrm{Ad}_{X^{-1}}\left( \left(\Omega - \frac{(R^\top x)^\times v}{|x|^2} \right)^\times, -  \frac{\left(x^\times Rv\right)^\times}{|x|^2}, - \frac{ x^\top Rv}{|x|^2} \right),
		\\
        &= \left( S^\top\left( \Omega - \frac{(R^\top x)^\times v}{|x|^2} \right)^\times S, - Q^\top \left( \frac{x^\times Rv}{|x|^2} \right)^\times Q ,  \frac{(Q^\top x/r)^\top (Q^\top R S)(S^\top v/r)}{|Q^\top x/r|^2}\right),  \\
        &= \left(\left( S^\top\Omega - \frac{S^\top R^\top x^\times R v}{|x|^2} \right)^\times, - \left( \frac{ Q^\top x^\times Rv}{|x|^2} \right)^\times ,  \frac{(Q^\top x/r)^\top (Q^\top R S)(S^\top v/r)}{|Q^\top x/r|^2}\right), \\
        &= \left(\left( S^\top\Omega - \frac{S^\top R^\top Q Q^\top x^\times Q Q^\top R S S^\top v}{|x|^2} \right)^\times ,  - \left( \frac{ Q^\top x^\times Q Q^\top R S S^\top v}{|x|^2} \right)^\times , \frac{(Q^\top x/r)^\top (Q^\top R S)(S^\top v/r)}{|Q^\top x/r|^2}\right), \\
        &= \left(\left( S^\top\Omega - \frac{(Q^\top R S)^\top (Q^\top x/r)^\times (Q^\top R S)(S^\top v/r)}{|Q^\top x/r|^2} \right)^\times,  - \left( \frac{ (Q^\top x/r)^\times (Q^\top R S) (S^\top v/r)}{|Q^\top x/r|^2} \right)^\times , \frac{(Q^\top x/r)^\top (Q^\top R S)(S^\top v/r)}{|Q^\top x/r|^2}\right),  \\
        &= f\left( \left(Q^\top R S, r^{-1}Q^\top x\right), \left(S^\top \Omega, r^{-1}S^\top v \right) \right), \\
        &= f\left( \phi\left((S,Q,r),(R, t) \right), \psi\left((S,Q,r),(\Omega, v) \right) \right).
	\end{align*}
    As required.
\end{proof}

\end{document}